\documentclass[oribibl,envcountreset,envcountsect]{llncs}
 
\usepackage{amssymb} 
\usepackage{amsmath} 
\usepackage[svgnames]{xcolor}

\definecolor{darkgreen}{rgb}{0,.35,0}
\definecolor{darkblue}{rgb}{0,0,.5}
\definecolor{darkred}{rgb}{.6,0,0}

\makeatletter
\spn@wtheorem{fact}{Fact}{\bfseries}{\itshape}
\makeatother

\usepackage[round]{natbib}

\newcommand{\Q}{{\mathfrak{S}}}
\renewcommand{\Q}{{\F(t)[\D;\delta]}}
\renewcommand{\S}{{\F[t][\D;\delta]}}
\newcommand{\D}{{\mathcal{D}}}
\newcommand{\CC}{{\mathbb{C}}}
\newcommand{\QQ}{{\mathbb{Q}}}
\newcommand{\NN}{{\mathbb{N}}}
\newcommand{\ZZ}{{\mathbb{Z}}}
\newcommand{\F}{{\mathsf{F}}}
\newcommand{\degD}{{\deg_{\D}}}
\newcommand{\degt}{{\deg_t}}
\newcommand{\nxn}{{n\times n}}

\newcommand{\lclm}{{\mbox{\upshape lclm}}}
\newcommand{\lcrm}{{\mbox{\upshape lcrm}}}
\newcommand{\gcld}{{\mbox{\upshape gcld}}}
\newcommand{\gcrd}{{\mbox{\upshape gcrd}}}
\newcommand{\uvec}{{\overrightarrow{u}}}
\newcommand{\vvec}{{\overrightarrow{v}}}
\newcommand{\zerovec}{{\overrightarrow{0}}}

\newcommand{\Hbar}{{\bar H}}
\newcommand{\diag}{{\mbox{\upshape diag}}}
\newcommand{\Plin}{{\widehat{P}}}
\newcommand{\Glin}{{\widehat{G}}}
\newcommand{\Alin}{{\widehat{A}}}
\newcommand{\inorm}[1]{{\|#1\|_\infty}}
\newcommand{\nonzero}{{\setminus\{0\}}}

\newcommand{\rowdeg}{{\mathop{rowdeg}}}

\numberwithin{equation}{section}

\begin{document}

\title{On computing the Hermite form of a matrix of differential
  polynomials}

\author{Mark Giesbrecht and Myung Sub Kim}

\institute{Cheriton School of Computer Science, University of
  Waterloo, Waterloo, Canada}


\maketitle

\begin{abstract}
  Given a matrix $A\in\F(t)[\D;\delta]^\nxn$ over the ring of
  differential polynomials, we show how to compute the Hermite form
  $H$ of $A$ and a unimodular matrix $U$ such that $UA=H$.  The
  algorithm requires a polynomial number of operations in $\F$ in
  terms of $n$, $\degD A$, $\degt A$.  When $\F=\QQ$ it require time
  polynomial in the bit-length of the rational coefficients as well.
\end{abstract}

\section{Introduction}

Canonical forms of matrices over principal ideal domains (such as
$\ZZ$ or $\F[x]$, for a field $\F$) have proven invaluable for both
mathematical and computational purposes. One of the successes of
computer algebra over the past three decades has been the development
of fast algorithms for computing these canonical forms. These include
triangular forms such as the Hermite form \citep{Hermite:1863}, low
degree forms like the Popov form \citep{Popov:1972}, as well as the
diagonal Smith form \citep{Smith:1861}.

Canonical forms of matrices over non-commutative domains, especially
rings of differential and difference operators, are also extremely
useful. These have been examined at least since \cite{Dickson:1923},
\cite{Wedderburn:1932}, and \cite{Jacobson:1943}. A typical domain
under consideration is that of differential polynomials. For our
purposes these are polynomials over a function field $\F(t)$ (where
$\F$ is a field of characteristic zero, typically an extension of
$\QQ$, or some representation of $\CC$). A differential indeterminate
$\D$ is adjoined to form the \emph{ring of differential polynomials}
$\F(t)[\D;\delta]$, which consists of the polynomials in $\F(t)[\D]$
under the usual addition and a non-commutative multiplication defined
such that $\D a=a\D+\delta(a)$, for any $a\in\F(t)$. Here
$\delta:\F(t)\to\F(t)$ is a \emph{pseudo-derivative}, a function such
that for all $a,b\in\F(t)$ we have
\[
\delta(a+b)=\delta(a)+\delta(b) ~~~\mbox{and}~~
\delta(ab)=a\delta(b)+\delta(a)b.
\] 
The most common derivation in $\F(t)$ takes $\delta(a)=a'$ for any
$a\in\F(t)$, the usual derivative of $a$, though other derivations
(say $\delta(t)=t$) are certainly of interest.

A primary motivation in the definition of $\F(t)[D;\delta]$ is that
there is a natural action on the space of infinitely differentiable
functions in $t$, namely the differential polynomial
\[
a_m\D^m+a_{m-1}\D^{m-1}+\cdots + a_1\D+a_0\in\F(t)[\D;\delta]
\]
acts as the linear differential operator
\[
a_m(t)\frac{d^m y(t)}{d
  t^m}+a_{m-1}(t)\frac{d^{m-1}y(t)}{d t^{m-1}}+
\cdots+a_1(t)\frac{d y(t)}{d t}+a_0(t)y(t)
\]
on a differentiable function $y(t)$.  Solving and analyzing systems of
such operators involves working with matrices over $\F(t)[\D;\delta]$,
and invariants such as the differential analogues of the Smith, Popov
and Hermite forms provide important structural information.

In commutative domains such as $\ZZ$ and $\F[x]$, it has been more
common to compute the triangular Hermite and diagonal Smith form (as
well as the lower degree Popov form, especially as an intermediate
computation). Indeed, these forms are more canonical in the sense of
being canonical in their class under multiplication by
unimodular matrices. Polynomial-time algorithms for the Smith and
Hermite forms over $\F[x]$ were developed by \cite{Kannan:1985}, with
important advances by \cite{KalKri87}, \cite{Villard:1995},
\cite{MuldersStorjohann:2003}, and many others.  One of the key
features of this recent work in computing normal forms has been a
careful analysis of the complexity in terms of matrix size, entry
degree, and coefficient swell. Clearly identifying and analyzing the
cost in terms of all these parameters has led to a dramatic drop in
both theoretical and practical complexity.

Computing the classical Smith and Hermite forms of matrices over
differential (and more general Ore) domains has received less
attention though normal forms of differential polynomial matrices have
applications in solving differential systems and control
theory. \cite{AbramovBronstein:2001} analyzes the number of reduction
steps necessary to compute a row-reduced form, while
\cite{BeckermannChengLabahn:2006} analyze the complexity of row
reduction in terms of matrix size, degree and the sizes of the
coefficients of some shifts of the input
matrix. \cite{BeckermannChengLabahn:2006} demonstrates tight bounds on
the degree and coefficient sizes of the output, which we will employ
here. For the Popov form, \cite{cheng:phd} gives an algorithm for
matrices of shift polynomials. Cheng's approach involves order bases
computation in order to eliminate lower order terms of Ore polynomial
matrices. A main contribution of \cite{cheng:phd} is to give an
algorithm computing the row rank and a row-reduced basis of the left
nullspace of a matrix of Ore polynomials in a fraction-free way. This
idea is extended in \cite{DaviesChengLabahn:2008} to compute Popov
form of general Ore polynomial matrices.  In
\cite{DaviesChengLabahn:2008}, they reduce the problem of computing
Popov form to a nullspace computation. However, though Popov form is
useful for rewriting high order terms with respect to low order terms,
we want a different normal form more suited to solving system of
linear diophantine equations. Since the Hermite form is upper
triangular it meets this goal nicely, not to mention the fact that it
is a ``classical'' canonical form.  In a slightly different vein,
\cite{Middeke:2008} has recently given an algorithm for the Smith
(diagonal) form of a matrix of differential polynomials, which
requires time polynomial in the matrix size and degree (but the
coefficient size is not analyzed).

In this paper, we first discuss some basic operations with polynomials
in $\F(t)[\D;\delta]$, which are typically written with respect to the
differential variable $\D$ as
\begin{equation}
  \label{eq:diffpol}
  f= f_0+f_1\D+f_2\D^2+\cdots+f_d\D^d,
\end{equation}
where $f_0,\ldots,f_d\in\F(t)$, with $f_d\neq 0$. We write $d=\degD f$
to mean the degree in the differential variable, and generally refer
to this as the \emph{degree} of $f$. Since this is a non-commutative
ring, it is important to set a standard notation in which the
coefficients $f_0,\ldots,f_d\in\F(t)$ are written to the left of the
differential variable $\D$. For $u,v\in\F[t]$ relatively prime, we can
define $\degt (u/v)=\max\{\degt u,\degt v\}$. This is extended to
$f\in\Q$ as in \eqref{eq:diffpol} by letting $\degt f=\max_i\{\degt
f_i\}$. We think of $\degt$ as measuring coefficient size or height.
Indeed, with a little extra work the bounds and algorithms in this
paper are effective over $\QQ(t)$ as well, where we also include the
bit-length of rational coefficients, as well as the degree in $t$, in
our analyses.

A matrix $U\in\Q^\nxn$ is said to be \emph{unimodular} if there exists
a $V\in\Q^\nxn$ such that $UV=I$, the $\nxn$ identity matrix. Note
that we do not employ the typical determinantal definition of a
unimodular matrix, as there is no easy notion of determinant for matrices
over $\Q$ (indeed, working around this deficiency suffuses much of our
work).

A matrix $H\in\Q^\nxn$ is said to be in \emph{Hermite form} if $H$ is
upper triangular, if every diagonal entry is monic, and every
off-diagonal entry has degree less than the diagonal entry below it.
As an example, the matrix
\[
\begin{pmatrix}
1+(t+2)\D+\D^{2} & 2+(2t+1)\D & 1+(1+t)\D\\
2t+t^{2}+t\D & 2+2t+2t^{2}+\D & 4t+t^{2}\\
3+t+(3+t)\D+\D^{2} & 8+4t+(5+3t)\D+\D^{2} & 7+8t+(2+4t)\D
\end{pmatrix}
\]
has Hermite form
\[
\begin{pmatrix}
2+{\it t+\D} & 1+2t & \frac{-2+t+2t^{2}}{2t}-\frac{1}{2t}\D\\
0 & 2+t+\D & 1+\frac{7t}{2}+\frac{1}{2}\D\\
0 & 0 & -\frac{2}{t}+\frac{-1+2t+t^{2}}{t}\D+\D^{2}
\end{pmatrix}.
\]
Note that the Hermite form may have denominators in $t$.  Also, while
this example does not demonstrate it, it is common that the degrees in
the Hermite form, in both $t$ an $\D$, are substantially larger than
in the input.

In this paper we will only concern ourselves with matrices in
$\Q^\nxn$ of full row rank, that is, matrices whose rows are
$\Q$-linear independent. For any matrix $A\in\Q^\nxn$, we show there exists a
unimodular matrix $U$ such that $UA=H$ is in Hermite form. This form
is canonical in the sense that if two matrices $A,B\in\Q^\nxn$ are
such that $A=PB$ for unimodular $P\in\Q^\nxn$ then the Hermite form of
$A$ equals the Hermite form of $B$.

The main contribution of this paper is an algorithm that, given a
matrix $A\in\Q^\nxn$ (of full row rank), computes $H$ and $U$ such 
that $UA=H$, which requires a polynomial number of $\F$-operations in
$n$, $\degD A$, and $\degt A$. It will also require time polynomial in 
the coefficient bit-length when $\F=\QQ$.

The remainder of the paper is organized as follows.  In Section 2 we
summarize some basic properties of differential polynomial rings and
present and analyze algorithms for some necessary basic operations.
In Section 3 we introduce a new approach to compute appropriate degree
bounds on the coefficients of $H$ and $U$. In Section 4 we present our
algorithm for computing the Hermite form of a matrix of differential
polynomials and analyze it completely.

\section{Basic structure and operations in $\F[t][\D;\delta]$}

In this section we discuss some of the basic structure of the ring
$\F(t)[\D;\delta]$ and present and analyze simple algorithms to do
some computations that will be necessary in the next section.

Some well-known properties of $\Q$ are worth recalling; see \linebreak
\cite{BroPet94} for an algorithmic presentation of this theory.  Given
$f,g\in\Q$, there is a degree function (in $\D$) which satisfies the
usual properties: $\degD (fg)=\degD f+\degD g$ and $\degD(f+g)\leq
\max\{\degD f,\degD g\}$.  $\Q$ is also a left and right principal
ideal ring, which implies the existence of a right (and left) division
with remainder algorithm such that there exists unique $q,r\in\Q$ such
that $f=qg+r$ where $\degD(r)<\degD(g)$.  This allows for a right (and
left) euclidean-like algorithm which shows the existence of a greatest
common right divisor, $h=\gcrd(f,g)$, a polynomial of minimal degree
(in $\D$) such that $f=uh$ and $g=vh$ for $u,v\in\Q$.  The GCRD is
unique up to a left multiple in $\F(t)\nonzero$, and there exist
co-factors $a,b\in\Q$ such that $af+bg=\gcrd(f,g)$.  There also exists
a least common left multiple $\lclm(f,g)$.  Analogously there exists a
greatest common left divisor, $\gcld(f,g)$, and least common right
multiple, $\lcrm(f,g)$, both of which are unique up to a right
multiple in $\F(t)$.

Efficient algorithms for computing products of polynomials are
developed in \cite{Hoeven:2002} and \cite{BCL:2008}, while fast
algorithms to compute the LCLM and GCRD, are developed in
\cite{Li:1997} and \cite{Li:1998}.  In this paper we will only need to
compute very specific products of the form $\D^kf$ for some $k\in\NN$.
We will work with differential polynomials in $\F[t][\D;\delta]$, as
opposed to $\F(t)[\D;\delta]$, and manage denominators separately.  If
$f\in\F[t][\D;\delta]$ is written as in \eqref{eq:diffpol}, then
$f_0,\ldots,f_d\in\F[t]$, and
\[
\D f = \sum_{0\leq i\leq d} f_i\D^{i+1} + \sum_{0\leq i\leq d}
f_i'\D^i\in\S,
\]
where $f_i'\in\F[t]$ is the usual derivative of $f_i\in\F[t]$.  Assume
$\degt f\leq e$.  It is
easily seen that $\degD(\D f)=d+1$, and $\degt(\D f)\leq e$.
The cost of computing $\D f$ is $O(de)$ operations in $\F$.
Computing $\D^kf$, for $1\leq k\leq m$ then requires $O(dem)$
operations in $\F$.

If $\F=\QQ$ we must account for the bit-length of the coefficients as
well.  Assuming our polynomials are in $\ZZ[t][\D;\delta]$ (which will
be sufficient), and are written as above, we have 
$f_i=\sum_{0\leq j\leq e} f_{ij} t^j$ for $f_{ij}\in\ZZ$.  We write 
$\inorm{f}=\max |f_{ij}|$ to capture the coefficient size of $f$.  It 
easily follows that $\inorm{\D f}\leq (e+1)\inorm{f}$, and so 
$\inorm{\D^mf}\leq (e+1)^m\inorm{f}$.

\begin{lemma}~\\[-\baselineskip]
\label{lem:shiftcost}
\begin{enumerate}
\item[(i)] Let $f\in\F[t][\D;\delta]$ have $\degD f=d$, $\degt f=e$, and
  let $m\in\NN$.  Then we can compute $\D^kf$, for $1\leq k\leq m$,
  with $O(dem)$ operations in $\F$.
\item[(ii)] Let $f\in\ZZ[t][\D;\delta]$.  Then $\inorm{\D^mf}\leq
  (e+1)^m\cdot\inorm{f}$, and we can compute $\D^if$, for $1\leq i\leq
  m$, with $O(dem\cdot (m\log e+\log\inorm{f})^2)$ bit operations.
\end{enumerate}
\end{lemma}

We make no claim that the above methods are the most efficient, and
faster polynomial and matrix arithmetic will certainly improve the
cost.  However, the above analysis will be sufficient, and these costs
will be dominated by others in the algorithms of later sections.

\section{Existence and degree bounds on the Hermite form}

In this section we prove the existence and uniqueness of the Hermite form
over $\F(t)[\D;\delta]$, and prove some important properties about
unimodular matrices  and equivalence over
this ring.  The principal technical difficulty is that there is no
natural determinant function with the properties found in 
commutative linear algebra.  The determinant is one of the main 
tools used in the analysis of essentially all fast
algorithms for computing the Hermite form $H$ and transformation matrix
$U$, and specifically two relevant techniques in established methods
by \cite{storjohann:ms1994} and \cite{KalKri87}. One approach might be
to employ the non-commutative determinant of \cite{Dieu:1943},
but this adds considerable complication.  Instead, we 
find degree bounds via established bounds on the row-reduced form.

\begin{definition}[Unimodular matrix] 
  Let $U\in\Q^\nxn$ and suppose there exists a $V\in\Q^\nxn$ such that
  $UV=I_n$, where $I_n$ is the identity matrix over $\Q^\nxn$. Then
  $U$ is called a \emph{unimodular} matrix over $\Q$.
\end{definition} 

\noindent
This definition is in fact symmetric, in that $V$ is also unimodular,
as shown in the following lemma (the proof of which is left to the reader).

\begin{lemma} 
  Let $U\in\Q^\nxn$ be unimodular such that there exists a
  $V\in\Q^\nxn$ with $UV=I_n$.  Then $VU=I_n$ as well.
\end{lemma} 

\begin{theorem} 
  \label{thm:unimod}
  Let $a,b\in\Q$.  There exists a unimodular matrix
  \[
  W=\begin{pmatrix}
    u & v\\
    s & t
  \end{pmatrix}
  \in\Q^{2\times 2}
  ~~\mbox{such that}~~
  W
  \begin{pmatrix}
    a\\ b
  \end{pmatrix}
  =
  \begin{pmatrix}
    g\\ 0
  \end{pmatrix},
  \]
  where $g=\gcrd(a,b)$ and $sa=-tb=\lclm(a,b)$. 
\end{theorem} 
\begin{proof} 
  Let $u,v\in\Q$ be the multipliers from the euclidean algorithm such
  that $ua+vb=g$.  Since $sa=-tb=\lclm(a,b)$, we know that
  $\gcld(s,t)=1$ (otherwise the minimality of the degree of the
  $\lclm$ would be violated). It follows that there exist $c,d\in\Q$
  such that $sc+td=1$.  Now observe that
  \[
  \begin{pmatrix}
    u & v\\
    s & t
  \end{pmatrix}
  \begin{pmatrix}
    ag^{-1}\ & c\\
    bg^{-1} & d
  \end{pmatrix}
  \begin{pmatrix}
    1 & -uc-vd\\
    0 & 1
  \end{pmatrix}
  =
  \begin{pmatrix}
    1 & uc+vd\\
    0 & 1
  \end{pmatrix}
  \begin{pmatrix}
    1 & -uc-vd\\
    0 & 1
  \end{pmatrix}
  = 
  \begin{pmatrix}
    1 & 0\\
    0 & 1
  \end{pmatrix}.
  \]
  Thus
  \[
  W^{-1} = 
  \begin{pmatrix}
      ag^{-1}\ & ag^{-1}(-uc-vd)+c\\
      bg^{-1}\ & bg^{-1}(-uc-vd)+d
    \end{pmatrix}
    = \begin{pmatrix}
      ag^{-1}\ & -a+c\\
      bg^{-1}\ &
      -b+d
    \end{pmatrix},
    \]
    so $W$ is unimodular.
  \qed
\end{proof} 
\begin{definition}[Hermite Normal Form] Let $H\in\Q^{n\times n}$ with
  full row rank. The matrix $H$ is in \emph{Hermite form} if $H$ is
  upper triangular, if every diagonal entry of $H$ is monic, and if
  every off-diagonal entry of $H$ has degree (in $\D$) strictly lower
  than the degree of the diagonal entry below it.
\end{definition} 

\begin{theorem} 
  \label{thm:sqherm}
  Let $A\in\Q^{n\times n}$ have row rank $n$. Then there exists a
  matrix $H\in\Q^\nxn$ with row rank $n$ in Hermite form, and a
  unimodular matrix $U\in\Q^\nxn$, such that $UA=H$.
\end{theorem} 
\begin{proof}
  We show this induction on $n$. The base case, $n=1$, is trivial and
  we suppose that the theorem holds for $n-1\times n-1$ matrices.
  Since $A$ has row rank $n$, we can find a permutation of the rows of $A$
  such that every principal minor of $A$ has full row rank.  Since
  this permutation is a unimodular transformation of $A$, we assume
  this property about $A$.  Thus, by the induction hypothesis, there
  exists a unimodular matrix $U_1\in\Q^{(n-1)\times (n-1)}$ such that
  \[ 
  \begin{pmatrix}
    &  &  &  & 0\\
    & U_1 &  &  & 0\\
    &  &  &  & \vdots\\
    &  &  &  & 0\\
    0 & 0 & \cdots & 0 & 1
  \end{pmatrix}
  \cdot A 
  = 
  \Hbar=
  \begin{pmatrix}
    \Hbar_{1,1} & \cdots & \cdots & * & * \\
    & \Hbar_{2,2} & \cdots  & * & *\\
\kern-25pt\raise8pt\hbox to 0pt{\vbox to 0pt{\huge 0}}
&  &  \ddots  & \vdots & \vdots\\
    &  &  &  \Hbar_{n-1,n-1} & *\\
    A_{n,1} & A_{n,2} & \cdots & A_{n,n-1} & A_{n,n}
  \end{pmatrix}
  \in\Q^\nxn,
  \] 
  where the $(n-1)$st principal minor of $\Hbar$ is in Hermite form. By
  Theorem~\ref{thm:unimod}, we know that there exists a unimodular
  matrix
  \[
  W=\begin{pmatrix}
      u_i & v_i\\
      s_i & -t_i
  \end{pmatrix}
  \in\Q^{2\times 2}
  ~~\mbox{such that}~~
  W\begin{pmatrix}
      \Hbar_{ii}\\
      A_{n,i}
    \end{pmatrix}
    =
  \begin{pmatrix}
    g_i\\
    0
  \end{pmatrix}\in\Q^{2\times 1}.
  \]
  This allows us to reduce $A_{n,1},\ldots,A_{n,n-1}$ to zero, and
  does not introduce any non-zero entries below the diagonal.  Also,
  all off-diagonal entries can be reduced using unimodular operations
  modulo the diagonal entry, putting the matrix into Hermite form.
  \qed
\end{proof}

\begin{corollary} 
  \label{cor:hermuniq}
  Let $A\in\Q^\nxn$ have full row rank.  Suppose $UA=H$ for unimodular
  $U\in\Q^\nxn$ and Hermite form $H\in\Q^\nxn$.  Then both $U$ and $H$
  are unique.
\end{corollary}
\begin{proof}
  Suppose $H$ and $G$ are both Hermite forms of $A$. Thus, there exist
  unimodular matrices $U$ and $V$ such that $UA=H$ and $VA=G$, and
  $G=WH$ where $W=VU^{-1}$ is unimodular.  Since $G$ and $H$ are upper
  triangular matrices, we know $W$ is as well.  Moreover, since $G$
  and $H$ have monic diagonal entries, the diagonal entries of $W$
  equal $1$. We now prove $W$ is the identity matrix. By way of
  contradiction, first assume that $W$ is not the identity, so there
  exists an entry $W_{ij}$ which is the first nonzero off-diagonal
  entry on the $i$th row of $W$. Since $i<j$ and since $W_{ii}=1$,
  $G_{ij}=H_{ij}+W_{ij}H_{jj}$. Because $W_{ij}\neq 0$, we see $\degD
  G_{ij}\geq\degD G_{jj}$, which contradicts the definition of the
  Hermite form.  The uniqueness of $U$ follows similarly.  \qed
\end{proof}

\begin{definition}[Row Degree] A matrix $T\in\Q^{n\times n}$ has
  row degree $\uvec\in (\NN\cup \{-\infty\})^n$ if the $i$th row of
  $T$ has degree $u_i$.  We write $\rowdeg \uvec$.
\end{definition} 
\begin{definition}[Leading Row Coefficient Matrix]  
  Let $T\in\Q^{n\times n}$ have $\rowdeg \uvec$. Set $N=\degD T$ and
  $S=\diag(\D ^{N-u_1},\ldots,\D ^{N-u_n})$.  We write
  \[
  ST=L\D ^N+\text{lower degree terms in $\D$},
  \] 
  where the matrix $L=LC_{row}(T)\in\F(t)^\nxn$ is called the \emph{leading
  row coefficient matrix} of $T$.
\end{definition} 
\begin{definition}[Row-reduced Form]  A matrix $T\in\Q^{m\times
    s}$ with rank $r$ is in row-reduced form if $rank$ $LC_{row}(T)=r$.
\end{definition} 

\begin{fact}[\cite{BeckermannChengLabahn:2006} Theorem~2.2]
  \label{fact:rowred}
  For any $A\in\Q^{m\times s}$ there exists a unimodular matrix
  $U\in\Q^{m\times m}$, with $T=UA$ having $r\leq\min\{ m,s\} $
  nonzero rows, $\rowdeg T\leq\rowdeg A$, and where the submatrix
  consisting of the r nonzero rows of T are row-reduced. Moreover, the
  unimodular multiplier satisfies the degree bound
  \[ 
  \rowdeg U\leq\vvec+(|\uvec|-|\vvec|-{\min_j}\{ u_j\}
  )\overrightarrow{e},
  \]
  where $\uvec:=\max(\zerovec,\rowdeg A)$, $\vvec:=\max(\zerovec,\rowdeg T)$,
  and $\overrightarrow{e}$ is the column vector with all entries equal to 1.
\end{fact}

\noindent The proof of the following is left to the reader.
\begin{corollary}
  \label{cor:rrid}
  If $A\in\Q^\nxn$ is a unimodular matrix then the row reduced form of
  $A$ is an identity matrix.
\end{corollary}

The following theorems provide degree bounds on $H$ and $U$.  We first
compute a degree bound of the inverse of $U$ by using the idea of
backward substitution, and then use the result of
\cite{BeckermannChengLabahn:2006} to compute degree bound of $U$.

\begin{theorem}
  \label{thm:hermunideg}
  Let $A\in \Q^\nxn$ be a matrix with $\degD A_{ij}\leq d$ and
  full row rank. Suppose $UA=H$ for unimodular matrix $U\in \Q^\nxn$ and
  $H\in \Q^\nxn$ in Hermite form. Then there exist a unimodular
  matrix $V\in\Q^\nxn$ such that $A=VH$ where $UV=I_n$ and $\degD
  V_{ij}\leq d$.
\end{theorem}
\begin{proof}
  We prove by induction on $n$. The base case is $n=1$. Since 
  $H_{11}=\gcrd(A_{11},\ldots,A_{n1})$, $\degD H_{11}\leq d$ and so 
  $\degD V_{i1}\leq d$ for $1\leq i\leq n$. Now, we suppose that our 
  claim is true for $k$ where $1<k<n$. Then we have to show that 
  $\degD V_{ik+1}\leq d$. We need to consider two cases:
  
  \noindent Case 1: $\degD V_{i,k+1}>\max(\degD V_{i1},\ldots,\degD V_{ik})$.
  Since
  \begin{align*}
  \degD H_{k+1,k+1} & \geq\max(\degD ~H_{1,k+1},\ldots,\degD H_{k,k+1}),\\
  \degD A_{i,k+1}&=\degD (V_{i,k+1}H_{k+1,k+1}),
  \end{align*}
  where $A_{i,k+1}=V_{i1}H_{1,k+1}+\cdots+V_{i,k+1}H_{k+1,k+1}$.  Thus, $\degD
  V_{i,k+1}\leq d$.

  \noindent Case 2: $\degD V_{i,k+1}\leq\max(\degD V_{i1},\ldots,\degD V_{ik})$.
  Thus, by induction hypothesis, $\degD V_{i,k+1}\leq d$.
  \qed
\end{proof}

\begin{corollary}
  \label{cor:degU}
  Let $A$, $V$, and $U$ be those in
  Theorem~\ref{thm:hermunideg}. Then $\degD U_{ij}\leq(n-1)d$.
\end{corollary}
\begin{proof}
  By Corollary~\ref{cor:rrid}, we know that the row reduced form of
  $V$ is $I_n$.  Moreover, since $I_n=UV$, we can compute the degree
  bound of $U$ by using Fact~\ref{fact:rowred}. Clearly,
  \[
  \vvec+(|\uvec|-|\vvec|- \min_j\{ u_j\}
  )\overrightarrow{e}\leq\vvec+(|\uvec|- \min_j\{ u_j\} )\overrightarrow{e},
  \] 
  where $\uvec:=\max(\zerovec,\rowdeg V)$ and
  $\vvec:=\max(\zerovec,\rowdeg I_n)=\zerovec$. Since the degree of
  each row of $V$ is bounded by $d$, $(|\uvec|- \min_j\{ u_j\} )\leq
  (n-1)d$.  Then, by Fact~\ref{fact:rowred}, $\rowdeg U\leq
  (n-1)d$. Therefore, $\degD U_{ij}\leq(n-1)d$.
  \qed
\end{proof}
\begin{corollary}
  \label{cor:Hnd}
  Let $H$ be same as that in Theorem~\ref{thm:hermunideg}. Then $\degD
  H_{ij}\leq nd$.
\end{corollary}
\begin{proof}
  Since $\degD U_{ij}\leq(n-1)d$ and $\degD A_{ij}\leq d$, $\degD
  H_{ij}\leq nd$.
  \qed
\end{proof}

\section{Computing Hermite forms by linear systems over $\F(t)$}

In this section we present our polynomial-time algorithm to compute
the Hermite form of a matrix over $\Q$.  We exhibit a variant of the
linear system method developed in \cite{KalKri87} and
\cite{storjohann:ms1994}. The approach of these papers is to reduce
the problem of computing the Hermite of matrices with (usual)
polynomial entries in $\F[z]$ to the problem of solving a linear
system equations over $\F$.  Analogously, we reduce the problem of
computing the Hermite form over $\S$ to solving linear systems over
$\F(t)$.  The point is that the field $\F(t)$ over which we solve is
the usual, commutative, field of rational functions.

For convenience, we assume that our matrix is over $\S$ instead of
$\Q$, which can easily be achieved by clearing denominators with a
``scalar'' multiple from $\F[t]$.  This is clearly a unimodular
operation in the class of matrices over $\Q$.

We first consider formulating the computation of the Hermite form a
matrix over $\Q$ as the solution of a ``pseudo''-linear system over
$\Q$ (i.e., a matrix equation over the non-commutative ring $\Q$).

\vspace*{-3pt}

\begin{theorem}
  \label{thm:hermsys}
  Let $A\in \S^\nxn$ have full row rank, with $\degD A_{i,j}\leq
  d$, and $(d_1,\ldots,d_n)\in\NN^n$ be given. Consider the system of
  equations $P A= G$, for $n\times n$ matrices for $P,G\in\Q$
  restricted as follows:

  \vspace*{-3pt}
  \begin{itemize}
  \item The degree (in $\D$) of each entry of $P$ is bounded by
    $(n-1)d+\max_{1\leq i\leq n} d_i$.
  \item The matrix $G$ is upper triangular, where every diagonal entry
    is monic and the degree of each off-diagonal entry is less than
    the degree of the diagonal entry below it.
  \item The degree of the $i$th diagonal entry of $G$ is $d_i$.
  \end{itemize}
  \vspace*{-2pt}
  Let $H$ be the Hermite form of $A$ and $(h_1,\ldots,h_n)\in\NN^n$ be
  the degrees of the diagonal entries of $H$. Then the following are
  true:
  \vspace*{-1pt}
  \begin{enumerate} 
  \item[(a)] There exists at least one pair $P,G$ as above with $PA=G$
    if and only if $d_i\geq h_i$ for $1\leq i\leq n$.
  \item[(b)] If $d_i=h_i$ for $1\leq i\leq n$ then $G$ is the Hermite
    form of $A$ and $P$ is a unimodular matrix.
  \end{enumerate}
\end{theorem} 
\begin{proof}
  The proof is similar to that of \cite{KalKri87}, Lemma 2.1.  Given a
  degree vector $(d_1,\ldots,d_n)$, we view $PA=G$ as a system of
  equations in the unknown entries of $P$ and $G$. Since $H$ is the
  Hermite form of $A$, there exist a unimodular matrix $U$ such that
  $UA=H$. Thus $PU^{-1}H=G$ and the matrix $PU^{-1}$ must be upper
  triangular since the matrices $H$ and $G$ are upper
  triangular. Moreover, since the matrix $PU^{-1}$ is in $\Q^\nxn$,
  and $G_{ii}=(PU^{-1})_{ii}\cdot H_{ii}$ for $1\leq i\leq n$, we know
  $d_i\geq h_i$ for $1\leq i\leq n$.  For the other direction, we
  suppose $d_i\geq h_i$ for $1\leq i\leq n$. Let $D=\diag(\D
  ^{d_1-h_1},\ldots,\D ^{d_n-h_n})$.  Then since $(DU)A=(DH)$, we can
  set $P=DU$ and $G=DH$ as a solution to $PA=G$, and the $i$th
  diagonal of $G$ has degree $d_i$ by construction. By
  Corollary~\ref{cor:degU}, we know $\degD U_{i,j}\leq(n-1)d$ and so
  $\degD P_{i,j}\leq(n-1)d+ \max_{1\leq i\leq n} d_i$.

  To prove (b), suppose $d_i=h_i$ for $1\leq i\leq n$ and that,
  contrarily, $G$ is \emph{not} the Hermite form of $A$. Since
  $PU^{-1}$ is an upper triangular matrix with ones on the diagonal,
  $PU^{-1}$ is a unimodular matrix.  Thus $P$ is a unimodular matrix
  and, by Corollary~\ref{cor:hermuniq}, $G$ \emph{is} the (unique)
  Hermite form of $A$, a contradiction.\qed
\end{proof}

\begin{lemma}
  \label{lem:Hlin}
  Let $A$, $P$, $(d_1,\ldots,d_n)$, and $G$ be as in
  Theorem~\ref{thm:hermsys}, and let \linebreak
  $\beta:=(n-1)d+\max_{1\leq i\leq n} d_i$.  Also, assume that $\degt
  A_{ij}\leq e$ for $1\leq i,j\leq n$.  Then we can express the system
  $PA=G$ as a linear system over $\F(t)$ as $\Plin\Alin=\Glin$ where
  \[
    \Plin \in\F(t)^{n\times n(\beta+1)}, \quad
    \Alin \in\F[t]^{n(\beta+1)\times n(\beta+d+1)}, \quad
    \Glin \in F(t)^{n\times n(\beta+d+1)}.
    \]
    Assuming the entries $\Alin$ are known while the entries of
    $\Plin$ and $\Glin$ are indeterminates, the system of equations
    from $\Plin\Alin=\Glin$ for the entries of $\Plin$ and $\Glin$ is
    linear over $\F(t)$ in its unknowns, and the number of equations
    and unknowns is $O(n^{3}d)$.  The entries in $\Alin$ are
    in $\F[t]$ and have degree at most $e$.
\end{lemma}
\begin{proof}
  Since $\degD P_{i,j}\leq \beta$, each entry of $P$ has at most
  $(\beta+1)$ coefficients in $\F(t)$ and can be written as
  $P_{ij}=\sum_{0\leq k\leq \beta} P_{ijk}\D^k$.  We let
  $\Plin\in\F(t)^{n\times n(\beta+1)}$ be the matrix formed from $P$
  with $P_{ij}$ replaced by the row vector
  $(P_{ij0},\ldots,P_{ij\beta})\in\F(t)$.

  Since $\degD P\leq \beta$, when forming $PA$, the entries in $A$ are
  multiplied by $\D^\ell$ for $0\leq\ell\leq\beta$, resulting in
  polynomials of degree in $\D$ of degree at most $\mu=\beta+d$.  Thus, we
  construct $\Alin$ as the matrix formed from $A$ with $A_{ij}$
  replaced by the $(\beta+1)\times (\mu+1)$ matrix whose $\ell$th row
  is
  \vspace*{-4pt}
  \[
  (A^{[\ell]}_{ij0}, A^{[\ell]}_{ij1}, \ldots, A^{[\ell]}_{ij\mu})
  ~~\mbox{such that}~~
  \D^\ell A_{ij} = A^{[\ell]}_{ij0}+A^{[\ell]}_{ij1}\D +
  \cdots +A^{[\ell]}_{ij\mu}\D^{\mu}.
  \]
  \vspace*{-15pt}

  \noindent
  Note that by Lemma~\ref{lem:shiftcost} we can compute
  $\D^{\ell}A_{i,j}$ quickly.

  Finally, we construct the matrix $\Glin$.  Each entry of $G$ has
  degree in $\D$ of degree at most $nd\leq n(\beta+d+1)$.  Thus,
  initially $\Glin$ is the matrix formed by $G$ with $G_{ij}$ replaced
  by
  \vspace*{-5pt}
  \[
  (G_{ij0},\ldots,G_{ij\mu})
  ~~~\mbox{where}~~~
  G_{ij}=G_{ij0}+G_{ij1}\D+\cdots+G_{ij\mu}\D^\mu.
  \]
  \vspace*{-16pt}

  \noindent
  However, because of the structure of the system we can fix values of
  many of the entries of $\Glin$ as follows.  First, since every
  diagonal entry of the Hermite form is monic, we know the
  corresponding entry in $\Glin$ is $1$.  Also, by
  Corollary~\ref{cor:Hnd}, the degree in $\D$ of every diagonal entry
  of $H$ is bounded by $nd$, and every off-diagonal has degree in $\D$
  less than that of the diagonal below it (and hence less than $nd$),
  and we can set all coefficients of larger powers of $\D$ to $0$ in
  $\Glin$.

  The resulting system $\Plin\Alin=\Glin$, restricted as above
  according to Theorem~\ref{thm:hermsys}, has $O(n^{3}d)$ linear
  equations in $O(n^3d)$ unknowns.  Since the coefficients in $\Alin$
  are all of the form $\D^\ell A_{ij}$, and since this does not affect
  their degree in $t$, the degree in $t$ of entries of $\Alin$ is the
  same as that of $A$, namely $e$.  \qed
\end{proof}

With more work, we believe the dimension of the system can be reduced
to $O(n^{2}d)\times O(n^{2}d)$ if we apply the techniques presented in
\cite{storjohann:ms1994} Section 4.3, wherein
the unknown coefficients of $\Glin$ are removed from the system.
See also \cite{Labhalla:1996}.

So far, we have shown how to convert the differential system over $\Q$
into a linear system over $\F(t)$. Also, we note, by
Theorem~\ref{thm:hermsys}, that the correct degree of the $i$th
diagonal entry in the Hermite form of $A$ can be found by seeking the
smallest non-negative integer $k$ such that $PA=G$ is consistent when
$\degD G_{j,j}=nd$ for $j=1,\ldots,i-1,i+1,\ldots,n$ and $k\leq\degD
G_{i,i}$. Using binary search, we can find the correct degrees of all
diagonal entries by solving at most $O(n\log(nd))$ systems.  We then
find the correct degrees of the diagonal entries in the Hermite form
of $A$, solving the system $PA=G$ with the correct diagonal degrees
gives the matrices $U$ and $H$ such that $UA=H$ where $H$ is the
Hermite form of $A$.

\begin{theorem}
  Let $A\in\F[t][\D;\delta]^\nxn$ with $\degD A_{ij}\leq d$ and
  $\degt A_{ij}\leq e$ for $1\leq i,j\leq n$. Then we can compute the
  Hermite form $H\in\Q$ of $A$, and a unimodular $U\in\S$ such that
  $UA=H$, with $O((n^{10}d^3+n^7d^2e)\log(nd))$ operations in $\F$
\end{theorem}
\begin{proof}
  Lemma~\ref{lem:Hlin} and the following discussion, above shows that
  computing $U$ and $H$ is reduced to solving $O(n\log(nd))$ systems
  of linear equations over $\F(t)$, each of which is $m\times m$ for
  $m=O(n^3d)$ and in which the entries have degree $e$.  Using
  standard linear algebra this can be solved with $O(m^4e)$ operations
  in $\F$, since any solution has degree at most $me$ (see
  \cite{GatGer03}).  A somewhat better strategy is to use the $t$-adic
  lifting approach of \cite{Dix82}, which would require $O(m^3+m^2e)$
  operations in $\F$ for each system, giving a total cost
  of $O((n^{10}d^3+n^7d^2e)\log(nd))$ operations in $\F$. \qed
\end{proof}
  
As noted above, it is expected that we can bring this cost down 
through a smaller system similar to that of \cite{storjohann:ms1994},
to a cost of $O((n^7d^2+n^5d^2e)\log(nd))$.   Nonetheless, the algorithm
as it is stated achieves a guaranteed polynomial-time solution.

It is often the case that we are considering differential systems over
$\QQ(t)[\D;\delta]$, where we must contend with growth in coefficients
in $\D$, $t$ \emph{and} in the size of the rational coefficients.
However, once again we may employ the fact that the Hermite form and
unimodular transformation matrix are solutions of a linear system over
$\QQ[t]$.  For convenience, we can assume in fact that our input is in
$\ZZ[t][\D;\delta]^\nxn$ (since the rational matrix to eliminate
denominators is unimodular in $\QQ(t)[\D;\delta]$).  There is some
amount of extra coefficient growth when going from $A$ to $\Alin$;
namely we take up to $nd$ derivatives, introducing a multiplicative
constant of size around $\min ((nd)!,e!)$.  In terms of the bit-length
of the coefficients, this incurs a multiplicative blow-up of only
$O(\ell\log(\ell))$ where $\ell=\min(nd,e)$.  It follows that we can
find the Hermite form of $A\in\QQ(t)[\D;\delta]^\nxn$ in time
polynomial in $n$, $\degt A_{ij}$, $\degD A_{ij}$, and
$\log\|A_{ij}\|$, the maximum coefficient length in an entry, for
$1\leq i,j\leq n$.  A modular algorithm, for example along the lines
of \cite{Li:1997}, would improve performance considerably, as might
$p$-adic solvers and a more careful construction of the linear system.

\section{Conclusions and Future Work}

We have shown that the problem of computing the Hermite form of a
matrix over $\F(t)[\D;\delta]$ can be accomplished in polynomial time.
Moreover, our algorithm will also control growth in coefficient
bit-length when $\F=\QQ$.  We have also shown that the degree bounds
on Hermite forms in the differential ring are very similar to the
regular polynomial case.  From a practical point of view our method is
still expensive.  Our next work will be to investigate more efficient
algorithms.  We have suggested ways to compress the system of
equations and to employ structured matrix techniques.  Also, the use
of randomization has been shown to be highly beneficial over $\F[t]$,
and should be investigated in this domain.  Finally, our approach
should be applicable to difference polynomials and more general Ore
polynomial rings.

\renewcommand\bibsection{\section*{References}}


\newcommand{\Gathen}{\relax}\newcommand{\Hoeven}{\relax}

\end{document}